\documentclass[journal]{IEEEtran}

\usepackage[utf8]{inputenc}
\usepackage[T1]{fontenc}
\usepackage{url, ifthen, cite}
\usepackage[cmex10]{amsmath}
\usepackage{graphicx, graphics, balance, dsfont}
\usepackage{amssymb, amsfonts, amsthm, algorithm2e, ulem}
\usepackage{epsfig, tikz, standalone, bm, makecell, tabu}

\newcommand{\PP}{\mathbb{P}}
\newcommand{\E}{\mathbb{E}}

\newtheorem{theorem}{Theorem}

\newtheorem{lemma}{Lemma}
\newtheorem{remark}{Remark}
\newtheorem{proposition}{Proposition}

\allowdisplaybreaks

\begin{document}
\title{Low-Complexity Random Rotation-based Schemes for Intelligent Reflecting Surfaces}

\author{Constantinos Psomas, \IEEEmembership{Senior Member, IEEE}, and Ioannis Krikidis, \IEEEmembership{Fellow, IEEE}

\thanks{C. Psomas and I. Krikidis are with the Department of Electrical and Computer Engineering, University of Cyprus, Nicosia, Cyprus (e-mail: \{psomas, krikidis\}@ucy.ac.cy). Preliminary results of this work have been presented at the IEEE International Symposium on Personal, Indoor and Mobile Radio Communications 2020, London, UK \cite{PIMRC}.

This work was co-funded by the European Regional Development Fund and the Republic of Cyprus through the Research and Innovation Foundation, under the projects POST-DOC/0916/0256 (IMPULSE) and INFRASTRUCTURES/1216/0017 (IRIDA). It has also received funding from the European Research Council (ERC) under the European Union's Horizon 2020 research and innovation programme (Grant agreement No. 819819).}}

\maketitle

\begin{abstract}
The employment of intelligent reflecting surfaces (IRSs) is a potential and promising solution to increase the spectral and energy efficiency of wireless communication networks. Despite their many advantages, IRS-aided communications have limitations as they are subject to high propagation losses. To overcome this, the phase rotation (shift) at each element needs to be designed in such a way as to increase the channel gain at the destination. However, this increases the system's complexity as well as its power consumption. In this paper, we present an analytical framework for the performance of random rotation-based IRS-aided communications. Under this framework, we propose four low-complexity and energy efficient schemes, based on a coding or a selection approach. Both of these approaches employ random phase rotations and require limited knowledge of channel state information. Specifically, the coding-based schemes use time-varying random phase rotations to produce an equivalent time-varying channel. On the other hand, the selection-based schemes select a partition of the IRS elements based on the received signal power at the destination. Analytical expressions for the achieved outage probability and energy efficiency of each scheme are derived. It is demonstrated that all schemes can provide significant performance gains as well as full diversity order.
\end{abstract}


\section{Introduction}
Despite the fact that the 5G-era has commenced, with its deployment in some countries, the challenge of how to connect billions of devices and satisfy their rate requirements still exists. Furthermore, the energy efficiency of such highly dense and highly connected wireless communication networks is another vital requirement of particular interest \cite{ZHANG}. A promising new technology which aims to address these issues is the so-called intelligent reflecting surfaces (IRSs), also known as reconfigurable intelligent surfaces \cite{EB,CH2}. An IRS consists of an array of passive elements embedded in a flat metasurface, where each element is reconfigurable and can alter the phase of the incident signal with the help of a dedicated controller \cite{CL}. Thus, through these software-controlled reflections of the signals, a smart and programmable wireless environment can be achieved \cite{MDR}. Their employment can provide many benefits such as extend the range of wireless communication systems, improve the spectral efficiency by means of their full-duplex operation as well as increase energy efficiency due to the passive operation of their elements.

As a result, IRS-aided communications has recently attracted substantial attention by the research community and the industry and has already been investigated under various different communication scenarios \cite{EMIL, RUI, CH, DING, QN, HANZO, EB2, RUI2, HANZO2, CHU, RUI3, HAN, SAAD}. Specifically, in \cite{RUI}, the authors study a single cell wireless system where a multi-antenna access point (AP) communicates with multiple users via an IRS; it is shown that the joint optimization of the active beamforming from the AP and the passive beamforming from the IRS can provide performance gains. A similar scenario is considered in \cite{QN}, where it is demonstrated that IRSs can outperform both half- and full-duplex amplify-and-forward relays. The implementation of two index modulation schemes, space shift keying and spatial modulation, in IRS-aided communications is studied in \cite{EB2}. It is shown that good spectral efficiency performance can be achieved even for low signal-to-noise ratio (SNR) values. The authors in \cite{DING}, consider IRS-assisted non-orthogonal multiple access (NOMA) communications and it is demonstrated that NOMA can benefit from the employment of IRSs. An upper bound for the ergodic spectral efficiency of an IRS-assisted system is evaluated in \cite{HAN}, and an optimal phase shift design is proposed to maximize the ergodic spectral efficiency. Moreover, the benefits from the employment of IRSs, in terms of physical layer security, are shown in \cite{CHU} and \cite{RUI3}. In particular, the work in \cite{CHU} designs the AP's transmit beamforming and the IRS's reflect beamforming, such that the transmit power is minimized subject to a secrecy rate constraint. On the other hand, in \cite{RUI3}, the authors jointly optimize the AP's transmit beamforming and the IRS's reflect beamforming in order to maximize the secrecy rate.

The energy efficiency of IRS deployments is investigated in \cite{CH}, where the proposed resource allocation methods achieved up to $300\%$ higher energy efficiency compared to the conventional multi-antenna amplify-and-forward relaying. On the other hand, compared to the conventional decode-and-forward relay system, an IRS achieves higher energy efficiency only when high data rates are required \cite{EMIL}. A stochastic geometry model with IRSs is presented in \cite{HANZO}, where the spatial randomness of users is taken into account. The derived analytical framework for the spectral and energy efficiency of the proposed model validates the gains from the employment of multiple IRSs. The uplink data rate in an IRS system is considered in \cite{SAAD}, where an asymptotic analysis is undertaken with imperfect channel estimation and correlated interference; it is shown that noise and interference from channel estimation errors become negligible as the number of elements increases. The implementation of IRSs has also been considered in the context of simultaneous wireless information and power transfer \cite{RUI2,HANZO2}. Specifically, the work in \cite{RUI2}, considers an IRS-aided wireless system with multiple information receivers and energy harvesters. By maximizing the weighted sum-power at the energy harvesters, it is demonstrated that IRS can enhance the performance. A similar scenario is considered in \cite{HANZO2}, where by maximizing the weighted sum-rate of the information receivers under certain energy harvesting constraints, it is shown that the existence of an IRS benefits the network.

Most of the aforementioned works, mainly focus on optimizing the incident signal's phase shifts at the IRS and assume knowledge of the channel state information (CSI). However, this corresponds to higher complexity and power consumption. Moreover, channel estimation in IRS-aided networks is challenging but can also be impractical in some cases, due to the limited resources of an IRS \cite{NADEEM}. Towards this direction, some efforts have been made, e.g. based on a minimum mean squared error (MMSE) approach \cite{NADEEM}, on deep reinforcement learning \cite{CH3} and on parallel factor decomposition \cite{WEI}. Motivated by this, in this paper, we present an analytical framework for the performance of random rotation-based IRS-aided communications. We propose four low-complexity and energy efficient techniques based on two approaches: a coding-based and a selection-based approach. Both approaches depend on random phase rotations and either do not require CSI (coding-based schemes) or have low CSI requirements (selection-based schemes) at the source. Specifically, the contribution of this work is threefold:
\begin{itemize}
\item A complete analytical framework for the performance of random rotation-based IRS schemes is presented. Building on this framework, we propose four low-complexity and well-connected schemes for IRS-aided communications. We derive analytical expressions for each scheme with respect to the outage probability and the energy efficiency. Furthermore, a diversity analysis is undertaken, where we provide the achieved diversity order and coding gain. Finally, we provide a detailed discussion on how these schemes can be implemented and how they compare with the coherent (beamforming) case in terms of perfect/imperfect CSI. Our results demonstrate that the proposed schemes can enhance the performance of IRS-aided communication systems both in terms of outage and energy efficiency.
\item We propose a random rotations coding-based (RRC) scheme, inspired by the rotate-and-forward protocol \cite{SHENG}, which produces an equivalent time-varying channel through time-varying random rotations. We show that RRC can achieve significant performance gains over a small number of channel uses, and provides full diversity order. Furthermore, we present a coding-based one-bit feedback (OBF) scheme, which adjusts the phase shift at each element according to a one-bit returned by the destination during a training period. It is demonstrated that, even for a short training period, the OBF scheme can improve the performance. We show that as the training period increases, the algorithm converges to the beamforming case.
\item Two selection-based schemes are proposed, which select a partition (sub-surface) of the IRS elements at each time slot based on the received signal power at the destination. In particular, the transmit diversity (TD) selects the sub-surface, which provides the highest achieved SNR at the destination. It is shown that the TD scheme provides full spatial diversity order and can substantially increase the energy efficiency. On the other hand, the adaptive transmit diversity (ATD) scheme selects a sub-surface, which achieves an SNR higher than a certain threshold. The ATD scheme is of lower complexity compared to the TD but can still achieve full diversity order and improve the performance.
\end{itemize}

The rest of this paper is organized as follows: Section \ref{sys_model} describes the considered system model and presents our main assumptions. In Section \ref{cbs} and Section \ref{sbs}, we present the proposed coding-based and selection-based schemes, respectively, together with their analytical expressions. Section \ref{comp} provides a discussion on the implementation issues of the proposed schemes together with a comparison. Numerical results are provided in Section \ref{numerical} and the paper concludes with Section \ref{conclusion}.

\textit{Notation}: Lower and upper case boldface letters denote vectors and matrices, respectively; $[\cdot]^\top$ denotes the transpose operator; $\Im\{z\}$ returns the imaginary part of $z$ and $\jmath = \sqrt{-1}$ denotes the imaginary unit; $\PP\{X\}$ and $\E\{X\}$ represent the probability and the expectation of $X$, respectively; $\mathds{1}_{X}$ is the indicator function of $X$ with $\mathds{1}_{X} = 1$ if $X$ is true and $\mathds{1}_{X} = 0$ otherwise; $K_M(\cdot)$ is the modified Bessel function of the second kind of order $M$, $\Gamma(\cdot)$ denotes the complete gamma function, $\log(\cdot)$ is the natural logarithm, and $\binom{n}{k} = \frac{n!}{(n-k)!k!}$ is the binomial coefficient.

\section{System Model}\label{sys_model}
Consider an IRS-aided network, where a source $S$ achieves communication with a destination $D$ through the employment of an IRS with $M$ reflecting elements, as shown in Fig. \ref{fig1}. The source and destination are equipped with a single antenna\footnote{The single antenna case refers to a low-complexity scenario and does not limit the contribution of the proposed framework; the multi-antenna case is left for future consideration.} \cite{EMIL} and a direct link between them is not available (e.g. due to high path-loss or heavy shadowing) \cite{QN,CH}. Assume a channel coherence period of duration $T$ (measured in channel uses). Then, a codeword
\begin{align}
\mathbf{x} \triangleq [x_1,x_2,\dots,x_T], x_t \in \mathbb{C}, 1\leq t\leq T,
\end{align}
is transmitted by the source over $T$ symbols time \cite{SHENG}. All wireless links are assumed to exhibit Rayleigh fading\footnote{Even though we consider Rayleigh fading, the proposed analytical framework is general and the extension to other fading models is straightforward as we need only to consider their probability distributions.} \cite{RUI,QN,EB2}; we define by $h_i$ and $g_i$ the fading coefficients from $S$ to the $i$-th IRS element and from the $i$-th IRS element to $D$, respectively. The fading coefficients remain constant during the $T$ transmissions but change independently every $T$ channel uses according to a circularly symmetric complex Gaussian distribution, i.e. $h_i \sim \mathcal{CN}(0,\sigma_h^2)$ and $g_i \sim \mathcal{CN}(0,\sigma_g^2)$; the variances capture both large- and small-scale fading effects. For the sake of simplicity, we define $\sigma^2 \triangleq \sigma_h^2\sigma_g^2$.

We assume that instantaneous knowledge of CSI at the source does not exist. At every time instant $t$, each element of the IRS, randomly rotates (shifts) the phase of the incident signal. Denote by
\begin{align}
  \mathbf{\Phi}_t = \text{diag}[\beta_1\exp(\jmath\phi_{t,1}) ~ \beta_2\exp(\jmath\phi_{t,2}) ~ \cdots ~ \beta_M\exp(\jmath\phi_{t,M})],
\end{align}
the diagonal matrix, where $\beta_i \in [0,1]$ is the reflection amplitude at the $i$-th IRS element and $\phi_{t,i}$ is the random phase shift, which is uniformly distributed in $[0,2\pi)$; unless otherwise stated, we consider $\beta_i = 1$, $\forall ~i$. Therefore, if the source transmits with a constant power $P$, the received signal at the destination $D$ at the $t$-th channel use can be written as
\begin{align}
r_t = \sqrt{P}\mathbf{h}^\top\mathbf{\Phi}_t\mathbf{g}x_t + n_t,
\end{align}
where $\mathbf{h} = [h_1 ~ h_2 ~ \cdots ~ h_M]^\top$, $\mathbf{g} = [g_1 ~ g_2 ~ \cdots ~ g_M]^\top$, and $n_t \sim \mathcal{CN}(0,\sigma_n^2)$ is the additive white Gaussian noise with variance $\sigma_n^2$. Then, the instantaneous SNR at the destination $D$ over the $t$-th transmission is
\begin{align}
\gamma_t = \frac{P}{\sigma_n^2} H_t,
\end{align}
where
\begin{align}
H_t = \Bigg\lvert \sum_{i=1}^M h_i g_i \exp(\jmath\phi_{t,i})\Bigg\rvert^2,
\end{align}
is the channel gain from the $M$ elements of the IRS.

\begin{figure}[t]\centering
  \includegraphics[width=0.8\linewidth]{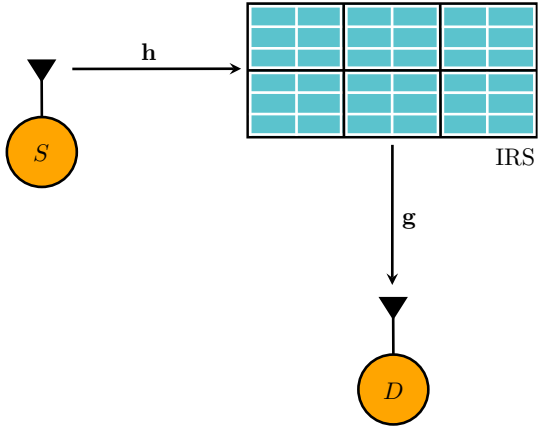}
  \caption{The considered IRS-aided communication network.}\label{fig1}
\end{figure}

Finally, we describe the main performance metrics considered for each scheme, namely, the outage probability, the energy efficiency and the diversity order. Let $\rho$ be a non-negative pre-defined threshold. Then,
\begin{align}
\Pi(\rho, T) = \PP\left\{\frac{1}{T} \sum_{t=1}^T \log_2\left(1+\gamma_t\right) < \rho\right\},
\end{align}
defines the achieved outage probability over $T$ channel uses. Moreover, the end-to-end energy efficiency, measured in bits per Joule, is written as
\begin{align}\label{ee}
\eta = \frac{\E\left\{\frac{1}{T} \sum_{t=1}^T \log_2\left(1+\gamma_t\right)\right\}}{P_c},
\end{align}
where $\E\left\{\frac{1}{T} \sum_{t=1}^T \log_2\left(1+\gamma_t\right)\right\}$ is the expected rate and $P_c$ is the system's total power consumption. Finally, if the outage probability of a scheme behaves like $\Pi(\rho, T) \approx G P^{-d}$ at high SNRs, then $d$ is the scheme's diversity order given by \cite{TSE}
\begin{align}
d = -\lim_{P\to\infty} \frac{\log(\Pi(\rho, T))}{\log(P)},
\end{align}
and $G$ is the coding gain \cite{TSE}.

\section{Coding-based IRS Schemes}\label{cbs}
In this section, we present the proposed coding-based schemes: the RRC scheme, which employs random phase shifts at each channel use, and the OBF scheme, which further implements a one-bit feedback protocol. The schemes are general and are not dependent on any particular coding design\footnote{The schemes can be implemented with any coding technique, e.g. Gaussian codebooks, which achieve capacity, polar codes or low-density parity check (LDPC) codes, which are more practical and can provide near-capacity performance.}; we omit any dependency on a specific coding technique, by considering as the main performance metric the outage probability instead of the error probability \cite{SHENG}.

We first provide a preliminary result in Proposition \ref{prop1}, which refers to the conventional non-coherent (random) case over one channel use, i.e. $T = 1$. This will assist in the derivation of the main analytical results of this paper.

\begin{proposition}\label{prop1}
The outage probability achieved by a random rotation of $M$ elements over one channel use is given by
\begin{align}\label{op}
\Pi(\rho,1) = 1-\frac{2}{\Gamma(M)}\left(\frac{\theta\sigma_n^2}{\sigma^2 P}\right)^{\frac{M}{2}} K_M\left(2 \sqrt{\frac{\theta\sigma_n^2}{\sigma^2 P}}\right),
\end{align}
where $\theta \triangleq 2^\rho - 1$.
\end{proposition}

\begin{proof}
See Appendix \ref{prop1_prf}.
\end{proof}

\subsection{Random Rotations Coding-based Scheme}\label{rrc}
The random phase shifts induced by the IRS elements at each channel use, ensure that the symbols are transmitted over $T$ independent channels. This time-scale fluctuations generate an artificial fast fading channel, which with an appropriate combination scheme at the receiver, e.g. maximal ratio combiner (MRC), converts the spatial diversity to time diversity \cite{SHENG}. It is worth mentioning that this scheme requires no CSI knowledge at the source. Note that the instantaneous channel gains between different channel uses are correlated, which makes the derivation of the outage probability challenging. As such, we present two approximations (lower bounds) in the following two theorems. In particular,
\begin{itemize}
  \item Theorem \ref{thm_rrc1} provides an approximate mathematical expression by assuming that the channel gains $H_t$ are mutually independent; it is proven analytically that, as $M$ increases, the correlation over different channel uses decreases.
  \item Theorem \ref{thm_rrc2} derives the outage probability by using the central limit theorem (CLT), which approximates $H_t$ as an exponential random variable; this results in a simpler analytical expression.
\end{itemize} 
We show in Section \ref{numerical} that both approximations are sufficient and appropriate to describe the proposed scheme's behavior.

\begin{theorem}\label{thm_rrc1}
The outage probability of the RRC scheme, under the independence assumption, is approximated by 
\begin{align}\label{op_rrc1}
\Pi_{\rm RRC}^{\rm IND}(\rho,T) &\approx \left(\frac{\sigma_n^2}{P}\right)^{T-1} \int_1^{\xi_T} \cdots \int_1^{\xi_2}\nonumber\\
&\quad\times\left(1-\frac{2}{\Gamma(M)}\left(\frac{\Theta}{\sigma^2}\right)^{M/2} K_M \left(2 \sqrt{\frac{\Theta}{\sigma^2}}\right)\right) \nonumber\\
&\quad\times \prod_{t=2}^T f_H\left(\frac{\sigma_n^2(w_t-1)}{P}\right) dw_2\cdots dw_T,
\end{align}
where $\xi_i \triangleq 2^{\rho T}/\prod_{t=i+1}^T w_t$, $2 \leq i \leq T$,
\begin{align}\label{Theta}
\Theta \triangleq \frac{\sigma_n^2}{P}\left(\frac{2^{\rho T}}{\prod_{t=2}^T w_t}-1\right),
\end{align}
and
\begin{align}\label{pdf}
f_H(h) = \frac{2}{\Gamma(M)}\frac{h^{(M-1)/2} }{\sigma^{M+1}}K_{M-1}\left(2\sqrt{\frac{h}{\sigma^2}}\right).
\end{align}
\end{theorem}

\begin{proof}
See Appendix \ref{thm_rrc1_prf}.
\end{proof}

It is clear that for the case $T=1$, Theorem \ref{thm_rrc1} corresponds to the exact analytical result in Proposition \ref{prop1}. Next, we provide the lemma below, which approximates the channel gain $H_t$ as an exponential random variable, and then state Theorem \ref{thm_rrc2}.

\begin{lemma}\label{lemma}
Under the CLT, the channel gain $H_t$ converges in distribution to an exponential random variable, with parameter $1/(\sigma^2 M)$.
\end{lemma}

\begin{proof}
See Appendix \ref{lemma_prf}.
\end{proof}

\begin{theorem}\label{thm_rrc2}
The outage probability of the RRC scheme, under the CLT, is approximated by 
\begin{align}\label{op_rrc2}
\Pi_{\rm RRC}^{\rm CLT}(\rho,T) &\approx \left(\frac{\sigma_n^2}{\sigma^2 MP}\right)^{T-1}\int_1^{\xi_T} \cdots \int_1^{\xi_2}\nonumber\\
&\quad\times\left(1-\exp\left(-\frac{\Theta}{\sigma^2 M}\right)\right)\nonumber\\
&\quad\times\prod_{t=2}^T \exp\left(-\frac{\sigma_n^2 (w_t-1)}{\sigma^2 MP}\right) dw_2\cdots dw_T,
\end{align}
where $\Theta$ and $\xi_i$ are given in Theorem \ref{thm_rrc1}.
\end{theorem}

\begin{proof}
By Lemma \ref{lemma}, the random variables $H_t$ are independent of $t$. Hence, the final result can be derived by following similar steps to the proof of Theorem \ref{thm_rrc1} with cumulative distribution function (CDF) $F_H(h)=1-\exp(-h/\sigma^2 M)$ and PDF $f_H(h)=\exp(-h/\sigma^2 M)/\sigma^2 M$.
\end{proof}

\begin{remark}
Even though the above approximations are based on the assumption that $M$ is large, they can also approximate small $M$ cases very well; this is verified in Section \ref{numerical}. Overall, Theorem \ref{thm_rrc1} provides a more accurate approximation for any value of $M$. On the other hand, Theorem \ref{thm_rrc2} uses exponential functions and thus can provide further system insights by assisting in the derivation of the diversity order and coding gain (see below). Finally, asymptotically ($M\to\infty$), both theorems produce the same results.
\end{remark}

We now turn our attention to the energy efficiency achieved by this scheme, as described by \eqref{ee}, which takes into account the achieved expected rate. Note that the expected rate of RRC is independent of $T$ since
\begin{align}
\E\left\{\frac{1}{T} \sum_{t=1}^T \log_2\left(1+\gamma_t\right)\right\} &= \frac{1}{T} \sum_{t=1}^T \E\left\{\log_2\left(1+\gamma_t\right)\right\}\nonumber\\
&= \E\left\{\log_2\left(1+\gamma_t\right)\right\}.
\end{align}
Thus, we can state the following.

\begin{proposition}\label{prop_rrc}
The expected rate achieved by the RRC scheme is
\begin{align}
R_{\rm RRC} = \frac{2}{\Gamma(M)} \int_0^\infty \left(\frac{\theta\sigma_n^2}{\sigma^2 P}\right)^{\frac{M}{2}} K_M\left(2 \sqrt{\frac{\theta\sigma_n^2}{\sigma^2 P}}\right) d\rho,
\end{align}
where $\theta \triangleq 2^\rho - 1$.
\end{proposition}

\begin{proof}
The result follows simply from the fact that the expectation of a non-negative random variable $X$ is given by $\E\{X\} = \int_{x>0} \PP\{X > x\} d x$. Therefore, $R_{\rm RRC} = \int_0^\infty \PP\left\{\log_2\left(1+\gamma_t\right) > \rho \right\} d\rho$, and the final expression is derived by using $\PP\left\{\log_2\left(1+\gamma_t\right) > \rho \right\} = 1- \Pi(\rho,1)$, where $\Pi(\rho,1)$ is given by Proposition \ref{prop1}.
\end{proof}

Then, from \eqref{ee}, we have that the energy efficiency achieved by the RRC scheme is
\begin{align}\label{ee_rrc}
\eta_{\rm RRC} = \frac{R_{\rm RRC}}{P/\xi + P_S + P_D +P_{\rm IRS}},
\end{align}
where $\xi$ is the amplifier's efficiency, whereas $P_S$, $P_D$ and $P_{\rm IRS}$ are the static power consumption at the source, destination and IRS, respectively \cite{CH}. The power consumption at the IRS depends on the number of activated elements, that is, $P_{\rm IRS} = MP_E$, where $P_E$ is the power consumed to operate a single element.

\subsection{One-bit Feedback Scheme}\label{obf}
We now present a coding-based scheme, which implements a distributed ascent algorithm \cite{RM} and is of low-complexity in terms of time and memory. The algorithm aims to achieve beamforming by adjusting the phase shift at each element based on a one-bit per channel use feedback protocol over a training period of duration $\tau \leq T$. The one-bit feedback from the destination to the IRS controller, dictates whether or not the change in the phase rotations has increased the received SNR compared to an SNR value $\gamma_0$ achieved at a previous channel use. Therefore, the only knowledge required is the set of phases that provided the highest channel gain at the destination. In other words, any set of phases that reduce the channel gain compared to a previous time instant are discarded. Specifically, the algorithm follows the steps below for the first $\tau \leq T$ channel uses
\begin{itemize}
\item At $t=1$, the IRS controller sets the initial phase shifts $\phi_{0,i} = \phi_{1,i} \in [0,2\pi)$, $\forall~i$, and the destination sets the initial value of $\gamma_0$, i.e. $\gamma_0 = \gamma_1 = \frac{P}{\sigma_n^2} \Big\lvert \sum_{i=1}^M h_i g_i \exp(\jmath\phi_{1,i})\Big\rvert^2$.
\item At each time instant $2 \leq t \leq \tau$, each element of the IRS rotates the phase of the received signal by using the following update step
$
\phi_{t,i} = \phi_{0,i} + \delta_{t,i},~ \forall i \in \{1,\dots,M\},
$
where $\delta_{t,i}$ is uniformly distributed in $[-\Delta,\Delta]$ and $\Delta \in (0,\pi]$ is the maximum step size.
\item If $\gamma_t > \gamma_0$, the destination returns a positive feedback and sets $\gamma_0 = \gamma_t$; otherwise, it sends a negative feedback and $\gamma_0$ remains unchanged. In turn, the IRS controller sets $\phi_{0,i} = \phi_{t,i}$ if it receives a positive feedback; otherwise, $\phi_{0,i}$ is not changed.
\end{itemize}

Then, at time instant $\tau+1$, the IRS fixes the phase rotations at $\phi_{0,i}$ for the remaining $T-\tau$ channel uses. Based on the above, at time instant $t$, the rotation angle of the $i$-th IRS element is
\begin{align}
\phi_{t,i} = \phi_{1,i} + \sum_{n=2}^t \delta_{n,i} \mathds{1}_{\gamma_n > \gamma_0},
\end{align}
while the channel gain at $t$ is
\begin{align}
H_t = \left|\sum_{i=1}^M h_i g_i \exp\left(\jmath \left(\phi_{1,i} + \sum_{n=2}^t \delta_{n,i} \mathds{1}_{\gamma_n > \gamma_0}\right)\right)\right|^2.
\end{align}

\begin{figure}[t]\centering
\includegraphics[width=\linewidth]{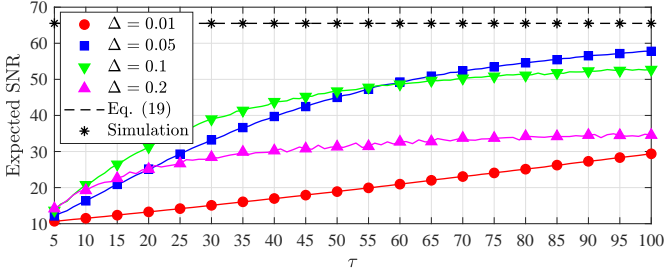}
\caption{Expected SNR versus duration of training period $\tau$; $\sigma^2 = 0$ dB, $P=0$ dB, $\sigma_n^2 = 0$ dB.}\label{fig2}
\end{figure}

It is important to note that the algorithm converges to the beamforming gain, i.e. phase alignment, for sufficiently large values of $\tau$ and $T$, regardless of the maximum step size $\Delta$; this will be shown in Section \ref{numerical}. However, for a low-complexity scenario, $\tau$ is fixed and the maximum channel gain that can be achieved, highly depends on the choice of $\Delta$. This is depicted in Fig. \ref{fig2}, where the expected SNR for the beamforming gain $H^* = \left(\sum_{i=1}^M \lvert h_i \rvert \lvert g_i \rvert\right)^2$is analytically given by
\begin{align}
&\E_{H^*}\left\{\frac{P}{\sigma_n^2} H^*\right\}\nonumber\\
&= \frac{P}{\sigma_n^2} \E\left\{\sum_{i=1}^M \lvert h_i \rvert^2 \lvert g_i \rvert^2 + 2\sum_{i=1}^{M-1} \sum_{j=i}^M \lvert h_i \rvert \lvert g_i \rvert \lvert h_j \rvert \lvert g_j \rvert\right\}\nonumber\\
&\stackrel{(a)}{=}\frac{P}{\sigma_n^2} \left(M\E\left\{\lvert h \rvert^2 \lvert g \rvert^2\right\} + 2\binom{M}{2} \E\left\{\lvert h \rvert \lvert g \rvert\right\}^2\right)\nonumber\\
&\stackrel{(b)}{=}\frac{P}{\sigma_n^2}\left(\sigma^2 M+ \binom{M}{2}\frac{\sigma^2 \pi^2}{8}\right),\label{ev_bf_snr}
\end{align}
where $(a)$ follows from the fact that $\lvert h_i \rvert \lvert g_i \rvert$ are mutually independent and $(b)$ follows from $\E\left\{\lvert h \rvert^2 \lvert g \rvert^2\right\} = \sigma_h^2\sigma_g^2 = \sigma^2$ and $\E\left\{\lvert h \rvert \lvert g \rvert\right\} = \sqrt{\sigma_h^2\pi}\sqrt{\sigma_g^2\pi}/4$.

The outage probability achieved by the proposed algorithm can be written as
\begin{align*}
\Pi_{\rm OBF}(\rho,\tau) = \PP\left\{\frac{1}{T} (R_{\rm TP} + R_{\rm BF}) < \rho\right\},
\end{align*}
where
\begin{equation}
R_{\rm TP} = \sum_{t = 1}^\tau \log_2\left(1+\frac{P}{\sigma_n^2}H_t\right),
\end{equation}
is the achieved sum-rate during the training period, and 
\begin{equation}
R_{\rm BF} = (T-\tau)\log_2\left(1+\frac{P}{\sigma_n^2}H_0\right),
\end{equation}
is the achieved sum-rate after completion of the training period, with constant phase rotations $\phi_{0,i}$ for $T-\tau$ channel uses; clearly, $\phi_{0,i} = \phi_{\tau,i}$ and $H_0 = H_\tau$. Furthermore, the energy efficiency achieved by the OBF scheme can be written as
\begin{align}
\eta_{\rm OBF} &= \frac{1}{T}\frac{\E\{R_{\rm TP} + R_{\rm BF}\}}{P/\xi + P_S + P_D +P_{\rm IRS}},\label{ee_obf}
\end{align}
where we assume that the power consumed by the one-bit feedback is negligible and so the total power consumption is equal to the one of the RRC scheme. In order to evaluate the expected rates in the above expression, we use the approximation of the expected rate for a channel with $b$ feedback bits by \cite{NJ}
\begin{align}\label{approx_rate}
&\mathbb{E}\left\{\log_2\left(1+\frac{P}{\sigma_n^2}H_t\right)\right\}\nonumber\\
&\approx \mathbb{E}\left\{\log_2\left(1+\frac{P}{\sigma_n^2} H^* \left(1-2^{-\frac{b}{M-1}}\right)\right)\right\},
\end{align}
where the $b$ feedback bits guarantee that each bit provides an increase in performance. However, based on the OBF scheme, not all feedback bits provide a performance gain. Hence, at the $t$-th time instant, we take $b = \kappa t$, where $\kappa \in (0,1]$ is a constant that is tuned numerically and determines the effectiveness of the algorithm in terms of $t$. Therefore, from \eqref{approx_rate}, we get
\begin{align}
&\mathbb{E}\left\{\log_2\left(1+\frac{P}{\sigma_n^2} H^* \left(1-2^{-\frac{\kappa t}{M-1}}\right)\right)\right\}\nonumber\\
&\leq \log_2\left(1+\frac{P}{\sigma_n^2} \mathbb{E}\left\{H^*\right\} \left(1-2^{-\frac{\kappa t}{M-1}}\right)\right)\nonumber\\
&= \log_2\left(1 + \frac{P}{\sigma_n^2}\left(M+ \binom{M}{2}\frac{\pi^2}{8}\right) \left(1-2^{-\frac{\kappa t}{M-1}}\right)\right),
\end{align}
which follows from the Jensen's inequality and the result in \eqref{ev_bf_snr}. By substituting the above expression in \eqref{ee_obf}, we get an analytical approximation for the energy efficiency.

\subsection{Diversity Analysis}
The RRC scheme virtually behaves as an $T\times M$ Rayleigh product channel \cite{SHENG2}. Since all the messages are sent by the source at each channel use, this is equivalent to using $T$ antennas over one time slot. Additionally, due to the independent random phase rotations at the IRS, the scheme can achieve diversity order equal to $\min(T,M)$ \cite{SHENG2}. By considering $P\to\infty$ and the two cases $M > T$ and $T > M$, we prove analytically in Appendix \ref{div_rrc} that
\begin{align}
d_{\rm RRC} = \min(T,M),
\end{align}
and
\begin{align}
G_{\rm RRC} &=\left(\frac{\sigma_n^2}{M}\right)^T (-1)^T\left(1 - 2^{\rho T} \sum_{t=0}^{T-1} \frac{(-1)^{t}}{t!} \log^t(2^{\rho T})\right),
\end{align}
are the achieved diversity order and coding gain of the RRC scheme, respectively.

Finally, as the OBF scheme employs the RRC scheme for the first $\tau$ channel uses, we can deduce that its diversity order is also $\min(\tau,M)$.

\section{Selection-based IRS Schemes}\label{sbs}
For the selection-based schemes, we consider $T=1$ and assume that the IRS is partitioned into $N$ non-overlapping sub-surfaces of $m$ elements, where $N$ is a divisor of $M$, i.e. $mN = M$. An example of the system model is illustrated in Fig. \ref{fig1} with $N=6$ and $m=6$, where the partitions are shown by the solid lines. We consider a closed-loop system, that is, we assume that there is knowledge of the received SNR power at the destination from each sub-surface via an error-free feedback scheme \cite{SLIM}. This can be implemented by a training period of duration $\tau$, where each sub-surface is turned on sequentially for a duration $\tau/N$; in other words, the reflection amplitude for those elements is set to $1$, otherwise it is set to zero\footnote{When an element is off, it acts as a conducting object and thus only structural-mode reflections are possible \cite{CAB}. This structural-mode component is deterministic and so its effect can be ignored.} \cite{NADEEM}. At the end of the training period, the destination feeds back to the IRS controller the index of the sub-surface which achieved the highest received signal strength (RSS).

\subsection{Transmit Diversity Scheme}
For the TD scheme, the IRS controller selects and activates, at each time slot, the sub-surface which achieves the highest SNR at the destination. As a result, the destination needs to feed back to the IRS $b_{\rm TD} = \lceil \log_2(N) \rceil$ bits. The outage probability achieved by the proposed TD scheme is given below.

\begin{proposition}\label{prop_td}
The outage probability of the TD scheme is
\begin{align}
\Pi_{\rm TD}(\rho) = \Pi(\rho,1)^N,
\end{align}
where $N$ is the number of sub-surfaces with $m$ elements and $\Pi(\rho,1)$ is the outage probability of a random selection given by Proposition \ref{prop1} with $M=m$.
\end{proposition}

\begin{proof}
Assume the ordering
\begin{align}
\gamma_{(1)} \geq \gamma_{(2)} \geq \dots \geq \gamma_{(N)},
\end{align}
where $\gamma_{(i)}$, $1 \leq i \leq N$, is the $i$-th highest receiver SNR at the destination from the $N$ sub-surfaces of the IRS. Then, using the distribution of ordered random variables, the outage probability is \cite{SLIM}
\begin{align}
\Pi_{\rm TD}(\rho) = \PP\{\log_2(1+\gamma) < \rho\}^N,
\end{align}
where $\PP\{\log_2(1+\gamma) < \rho\}$ is the outage probability when $T=1$ given by Proposition \ref{prop1} and the result follows.
\end{proof}

In what follows, we evaluate the energy efficiency of the TD scheme.

\begin{proposition}\label{prop_rtd}
The expected rate achieved by the TD scheme is
\begin{align}\label{rate_td}
R_{\rm TD} = N\int_0^\infty \log_2\left(1+\frac{P}{\sigma_n^2}h\right) F_H(h)^{N-1} f_H(h) dh,
\end{align}
where $f_H(h)$ and $F_H(h)$ are given by \eqref{pdf} and \eqref{cdf}, respectively.
\end{proposition}

\begin{proof}
See Appendix \ref{prop_rtd_prf}.
\end{proof}

Therefore, the energy efficiency achieved by the TD scheme is
\begin{align}\label{ee_td}
\eta_{\rm TD} = \frac{R_{\rm TD}}{P/\xi + P_S + P_D +P_{\rm IRS}},
\end{align}
where the power consumption parameters are defined as before but with $P_{\rm IRS} = m P_E$, since only $m$ elements are activated at each time slot. It is clear, that $\eta_{\rm SB} = \eta_{\rm CB}$ when $N=T=1$. On the other hand, for $N = T > 1$, the denominator of \eqref{ee_td} is always less that the one of \eqref{ee_rrc}, since $m < M$. Therefore, as $M$ increases, the TD scheme becomes more energy efficient. Note that for $N=1$, Proposition \ref{prop_rtd} provides the expected rate for a randomly selected sub-surface.

\subsection{Adaptive Transmit Diversity Scheme}
We now consider the ATD scheme, where the IRS selects a sub-surface which achieves an SNR at the destination of at least $\psi$ \cite{YANG}. Initially, the IRS activates a random sub-surface and the destination feeds back one bit, representing whether or not the received signal achieved the threshold $\psi$. In case of a positive feedback, the IRS selects that sub-surface for the remaining communication period; otherwise, the same process is repeated with a different sub-surface. If the first $N-1$ sub-surfaces do not satisfy the selection criterion, then the IRS selects the $N$-th sub-surface, regardless of its achieved SNR.

Without loss of generality, assume that the IRS activates the sub-surfaces in an order which achieve SNRs $\gamma_1, \gamma_2, \dots, \gamma_{N-1}$. Therefore, the average number of feedback bits needed are
\begin{align}
b_{\rm ATD} &= 1 + \sum_{i=1}^{N-2} \PP\{\log_2(1+\gamma_i) < \psi\}\nonumber\\
&= 1 + (N-2) \Pi(\psi,1),
\end{align}
where $\Pi(\psi,1)$ is given by Proposition \ref{prop1}. The outage probability for this scheme, is given by the following proposition.

\begin{proposition}\label{prop_atd}
The outage probability of the ATD scheme is
\begin{align}\label{atd}
\Pi_{\rm ATD}(\rho,\psi) &= \Pi(\psi,1)^{N-1}\Pi(\rho,1)\nonumber\\
&\quad+ \mathds{1}_{\rho > \psi}(\Pi(\rho,1)-\Pi(\psi,1)) \sum_{k=0}^{N-2} \Pi(\psi,1)^k,
\end{align}
where $N$ is the number of sub-surfaces with $m$ elements and $\Pi(\rho,1)$ is the outage probability of a random selection given by Proposition \ref{prop1} with $M=m$.
\end{proposition}

\begin{proof}
See Appendix \ref{prop_atd_prf}.
\end{proof}

We can observe from Proposition \ref{prop_atd}, that an increase in $N$ is always beneficial for the case $\rho \leq \psi$. However, when $\rho > \psi$, the second term in \eqref{atd} increases with $N$. In addition, the case $\psi = \rho$, describes the TD scheme. Now, for the expected rate of this scheme, we can write
\begin{align}
R_{\rm ATD} &= \sum_{k=0}^{N-2} \Pi(\psi,1)^k \int_\psi^\infty \log_2\left(1+\frac{P}{\sigma_n^2}h\right) f_H(h) dh\nonumber\\
&\quad+ \Pi(\psi,1)^{N-1} \int_0^\infty \log_2\left(1+\frac{P}{\sigma_n^2}h\right) f_H(h) dh,
\end{align}
where $f_H(h)$ is given by \eqref{pdf}. We omit the proof for brevity as it follows a similar approach as the proof of Proposition \ref{prop_rrc}.

Note that this scheme could be generalized, in the sense that the IRS could stop after activating $K$ sub-surfaces, with $K \leq N-1$. The considered case provides the upper bound in terms of performance, but our analysis could be generalized by simply setting $N = K+1$. Finally, the energy efficiency $\eta_{\rm ATD}$ of the ATD scheme is simply given by \eqref{ee_td} but with expected rate $R_{\rm ATD}$ provided above.

\subsection{Diversity Analysis}
We now derive the diversity order and coding gain of the selection-based schemes; the proofs can be found in Appendix \ref{div_sbs}. Specifically, the TD scheme, achieves full spatial diversity order, i.e. $d_{\rm TD} = N$, as expected. Moreover, its achieved coding gain is equal to
\begin{align}\label{ag_td}
G_{\rm TD} = \left(\frac{\sigma_n^2}{m}\right)^N (2^\rho-1)^N.
\end{align}
To compare $G_{\rm TD}$ with $G_{\rm RRC}$, we need to consider the case $T = N$, i.e. equal diversity order. If $m = M/N$, then it is clear that $G_{\rm TD} > G_{\rm RRC}$. However, if each sub-surface employs $M$ elements then $G_{\rm TD} < G_{\rm RRC}$; in this case, the selection-based scheme employs more elements ($MN$ in total) but activates the same number as the coding-based scheme.

Finally, the diversity order of the ATD scheme depends on whether or not $\rho \leq \psi$. In particular, if $\rho \leq \psi$, it is clear from Proposition \ref{prop_atd} that the diversity order is $N$ with a coding gain
\begin{align}
G_{\rm ATD} = \left(\frac{\sigma_n^2}{m}\right)^N (2^\psi-1)^{N-1}(2^\rho-1) \geq G_{\rm TD},
\end{align}
where equality holds for $\psi = \rho$. On the other hand, if $\rho > \psi$, the second term of \eqref{atd} dominates and so the achieved diversity is one.

\subsection{Limiting Distribution}
Next, we consider the asymptotic behavior of the TD scheme as $N$ increases. Clearly, when $N \to \infty$ then $M \to \infty$, which corresponds to a massive multiple-element configuration and is a case of practical interest \cite{EB, QN}.

Now, based on extreme value theory, when the selection is done over a large number of sub-surfaces, the limiting distribution of the largest order statistic can be one of three domains of attraction, namely, the Fr\'echet, the Weibull and the Gumbel distribution \cite{SLIM}. In our case, using Lemma \ref{lemma}, we can easily prove that the parent distribution satisfies
\begin{equation}
\lim_{x\to\infty} \frac{1-F_H(x)}{f_H(x)} = c,
\end{equation}
where $c > 0$ is a constant. As a result, $\Pi_{\rm TD}(\rho)$ converges to a Gumbel distribution, i.e.
\begin{equation}
\Pi_{\rm TD}(\rho) = G\left(\frac{\theta\sigma_n^2/P-b_N}{a_N}\right),
\end{equation}
where $G(x)$ is given by
\begin{equation}\label{gumbel}
G(x) = \exp(-\exp(-x)), -\infty < x < \infty.
\end{equation}
Moreover, $a_N$ and $b_N$ are normalizing constants satisfying the following condition
\begin{equation}
\lim_{N \to \infty} F_H(a_N x + b_N) = G(x),
\end{equation}
where $F_H(\cdot)$ is given by Lemma \ref{lemma}. These constants can be computed by solving the following
\begin{align}
1-F_H(b_N) = \frac{1}{N},
\end{align}
and
\begin{align}
1-F_H(a_N + b_N) &= \frac{1}{eN},
\end{align}
where $e$ is Euler's number, which results in $a_N = \sigma^2 m$ and $b_N = \sigma^2 m\log(N)$. Therefore, we have
\begin{align}
\Pi_{\rm TD}(\rho) &= \exp\left(-\exp\left(-\frac{\theta\sigma_n^2/P-\sigma^2m\log(N)}{\sigma^2m}\right)\right)\nonumber\\
&= \exp\left(-N\exp\left(-\frac{\theta\sigma_n^2}{\sigma^2mP}\right)\right).
\end{align}

\begin{table*}[t]\caption{Summary of IRS Schemes}\centering{\tabulinesep=0.5mm
		\begin{tabu}{|c||c|c|c|c|c|}\hline
			\textbf{~} & \textbf{CT} & \textbf{RRC} & \textbf{OBF} & \textbf{TD} & \textbf{ATD}\\\hline
			\textbf{Diversity} & $M$ & $\min(T,M)$ & $\min(\tau,M)$ & $N$ & \makecell{$N \text{~if~} \rho \leq \psi$,\\ $1 \text{~if~} \rho > \psi$}\\\hline
			\textbf{Signaling (bits)} & $k M, k \in \mathbb{Z}^+$ & $0$ & $\tau$ & $\lceil \log_2(N) \rceil$ & $1+(N-2)\Pi(\psi,1)$\\\hline
			\textbf{Active elements} & $M$ & $M$ & $M$ & $M/N$ & $M/N$\\\hline
			\textbf{Training phase} & Yes & No & No & Yes & Yes\\\hline
			\textbf{Pre-log factor}& $1-\tau/T$ & $1$ & $1$ & $1-\tau/T$ & See Section \ref{comp}\\\hline
			$\sigma_e^2$ & $(1+\frac{\tau}{M} P_o)^{-1}$ & $-$ & $-$ & $-$ & $-$\\\hline
			\textbf{Phase-shift values} & Continuous & Continuous & Continuous & Discrete & Discrete\\\hline
	\end{tabu}}\label{table}
\end{table*}

\section{Implementation Issues \& Comparison}\label{comp}
In this section, we provide a discussion regarding the implementation issues of the proposed schemes and present a brief comparison between their benefits and capabilities. We will consider as the main benchmark the case of coherent transmission (CT), i.e. the case where the phases are aligned, with perfect and imperfect CSI. The implementation of CT (beamforming) depends on the accuracy of a channel estimation protocol, which takes place before every communication phase \cite{NADEEM}. Channel estimation in IRS-aided communications is a non-trivial task, due to the fact that the IRS passively reflects the incident signals and does not have any signal processing capabilities. For our analysis, we consider a channel estimation protocol based on an MMSE approach \cite{NADEEM}. Assume a training phase of duration $\tau$, divided into $M$ sub-phases (one for each IRS element) of duration $\tau/M$. During the $i$-th sub-phase, the IRS switches the $i$-th element on (i.e. $\beta_i = 1$) while keeping all the other elements off (i.e. $\beta_j = 0$, $j \neq i$). Then, the destination transmits a pilot symbol $x_i$, where $|x_i|^2 = P_o, \forall ~i$ and $x_i, x_j, i\neq j$ are mutually orthogonal. In this way, the source can estimate the cascaded channel through the $i$-th element. Based on this approach, we can state the following theorem.

\begin{theorem}\label{propc}
The outage probability achieved by the CT (beamforming) with $M$ elements under imperfect CSI is given by
\begin{align}
\Pi_{\rm CT}(\rho) = \frac{1}{\pi} \int_0^\infty \Im\Bigg\{\phi(t) &\sum_{i=2M}^\infty (-1)^{i+1}\frac{t^{i-1}}{i!}\nonumber\\
&\times \left(\jmath \sqrt{\frac{\theta(\sigma_n^2+\sigma_e^2)}{P}}\right)^i\Bigg\} dt,
\end{align}
with
\begin{align}\label{cfH}
\phi(t) = \left(\frac{4\sqrt{t^2 s + 4} + 2\jmath t \sqrt{s}\left(\pi + 2\jmath \sinh^{-1}\left(\frac{t\sqrt{s} }{2}\right)\right)}{\left(t^2 s+4\right)^{3/2}}\right)^M,
\end{align}
where $s=\sigma^2(1-\sigma_e^2)$ and $\sigma_e^2 = 1/(1+\tau P_o/M)$.
\end{theorem}

\begin{proof}
See Appendix \ref{propc_prf}.
\end{proof}

The perfect CSI case corresponds to $\sigma_e^2 = 0$. It follows that the energy efficiency of CT is
\begin{align}
\eta_{\rm CT} = \frac{(1-\tau/T)R_{\rm CT}}{(1-\tau/T)(P/\xi + P_{\rm IRS}) + \tau M P_o/T + P_S + P_D},
\end{align}
where $R_{\rm CT} = \mathbb{E}\{\log_2(1+\gamma_{\rm CT})\} = \int_0^\infty (1-\Pi_{\rm CT}(\rho)) d\rho$ is the achieved rate; note that the factors $1-\tau/T$ and $\tau M P_o/T$ take into account the rate loss and power consumption due to the training process, respectively.

We now provide a discussion as to how our proposed schemes compare to the above approach. In particular, the RRC scheme does not require a training phase since no CSI is needed. Therefore, it is not affected by any rate losses (pre-log factor is $1$), estimation errors ($\sigma_e^2 = 0$) or further power consumption due to channel training ($P_o = 0$). The OBF scheme operates in a similar way, but requires signaling (feedback of $\tau$ bits) to implement. The TD scheme has a training phase of duration $\tau$ and, similar to the CT scheme, each sub-surface is turned on sequentially in order to select the one with the highest RSS. However, in contrast to the CT scheme, it requires less power consumption since there are $N < M$ sub-phases (i.e. $\tau N P_o < \tau M P_o$) and there are no estimation errors ($\sigma_e^2 = 0$) due to the RSS-based selection approach. Finally, the ATD scheme behaves similarly to the TD scheme but the pre-log factor varies depending on which sub-surface achieves the required threshold $\psi$.

The above comparison is summarized in Table \ref{table}. Note that we do not take into account the costs (in terms of rate and power consumption) due to signaling. However, it is clear from the table, that the CT scheme requires a larger number of bits at each sub-phase compared to our proposed schemes. In addition, the random rotation of the phases at each IRS element guarantees that the complexity at the IRS is kept low, independently of the phases' resolution. However, the CT, RRC and OBF scheme require continuous phase shifts whereas the selection-schemes do not have this restriction. Finally, it is important to point out, that a coding-based scheme can easily be jointly implemented with a selection-based scheme; for example, the RRC scheme could be employed with the TD scheme on the selected sub-surface. We consider them separately so as to emphasize the benefits of each approach but their joint consideration is a simple extension of the derived analytical results.

\section{Numerical Results}\label{numerical}

We now validate our theoretical analysis and main analytical assumptions with computer simulations and show the benefits of our proposed schemes. For the sake of presentation, we consider $\rho = 1$ bps/Hz, $\sigma^2 = 0$ dB, $\sigma_n^2 = 0$ dB, $\xi = 1.2$, $P_E = 10$ dBm, $P_D = 10$ dBm and $P_S = 9$ dBW \cite{CH}. Moreover, all of the proposed schemes are compared with the conventional non-coherent case (Proposition \ref{prop1}), i.e. $T=1$, $N=1$, and with the CT case (Theorem \ref{propc}). Unless otherwise stated, lines correspond to theoretical results whereas markers correspond to simulation results.

\begin{figure}[t]\centering
  \includegraphics[width=\linewidth]{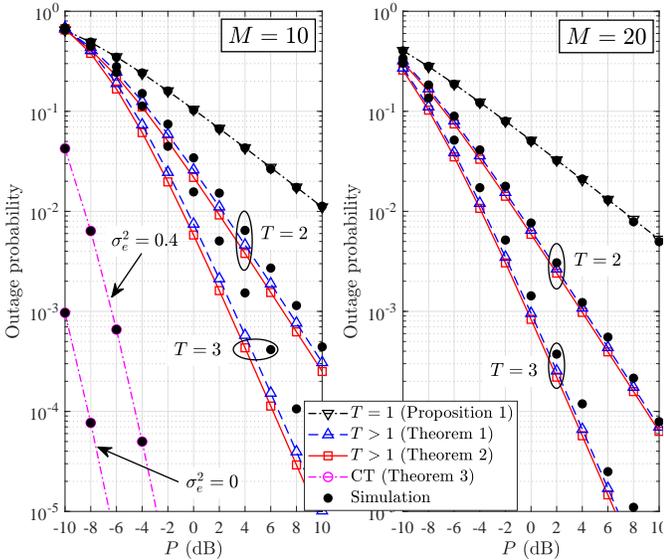}
  \caption{Outage probability versus $P$ for the RRC scheme.}\label{fig3}
\end{figure}

Fig. \ref{fig3} depicts the outage probability achieved by the RRC scheme in terms of the transmit power $P$, the number of channel uses $T$ and the number of reflecting elements $M$. As expected, the performance is improved with an increase of $M$. Moreover, and most importantly, increasing the number of channel uses $T$ provides significant gains to the outage probability. It should be highlighted that the massive benefits of this scheme are evident from $T=2$, which is the smallest number of channel uses the scheme can employ. Indeed, the outage probability is reduced around $96\%$ and $99\%$ for $T=2$ and $T=3$, respectively, compared to the conventional case with $M=10$. We can observe that the scheme provides full diversity order $T$, since $T < M$, as deduced by our analysis. Fig. \ref{fig3} also illustrates the performance of the CT scheme for $\sigma_e^2 = 0$ and $\sigma_e^2 = 0.3$. The case $\sigma_e^2 = 0$ provides the lower bound, as expected. On the other hand, in the presence of errors, the performance significantly deteriorates. It is clear that, by increasing $T$, our scheme will get closer to the performance of the CT. Finally, the figure validates the considered assumptions and approximations of our theoretical study. Specifically, the simulations perfectly match the theoretical results of Proposition \ref{prop1} and Theorem \ref{propc}. Furthermore, for $T>1$, the expressions of Theorem \ref{thm_rrc1} and Theorem \ref{thm_rrc2}, approximate the achieved outage probability exceptionally well even for small $M$ and the approximations become tighter as $M$ increases.

Fig. \ref{fig4} shows the achieved outage probability of the OBF scheme with $M = 10$ and for the cases $\tau = 40$ with $\delta = 0.1$ and $\tau = 100$ with $\delta = 0.05$. Note that, in contrast to Fig. \ref{fig3}, we consider larger values for $T$ since $T \geq \tau$. The performance of the OBF scheme is compared to the CT case, which is the best scenario that can be achieved; the theoretical and simulation results for the CT case agree, which verifies our analysis. As an additional benchmark, we consider the transmit beamforming through the equivalent multiple-input single-output (MISO) cascaded channel. For this case, we follow a $1$-bit channel learning approach based on the analytic center cutting plane method (ACCPM) \cite[Sec. IV]{RUI4}. Also, we show the perfect MISO beamforming case, i.e. when $\gamma_t = \frac{P}{\sigma_n^2} \sum_{i=1}^M \lvert h_i \rvert^2 \lvert g_i \rvert^2$. It can be observed that, for both training periods, the OBF scheme outperforms the ACCPM-based and the perfect MISO beamforming. Moreover, as both $\tau$ and $T$ increase, the OBF gets closer to the CT case, which justifies our claims in Section \ref{obf}. However, for a fixed $\tau$, the outage probability converges to a lower bound as $T$ increases. Thus, for a low-complexity scenario, $T$ does not need to be much larger than $\tau$.

\begin{figure}[t]\centering
	\includegraphics[width=\linewidth]{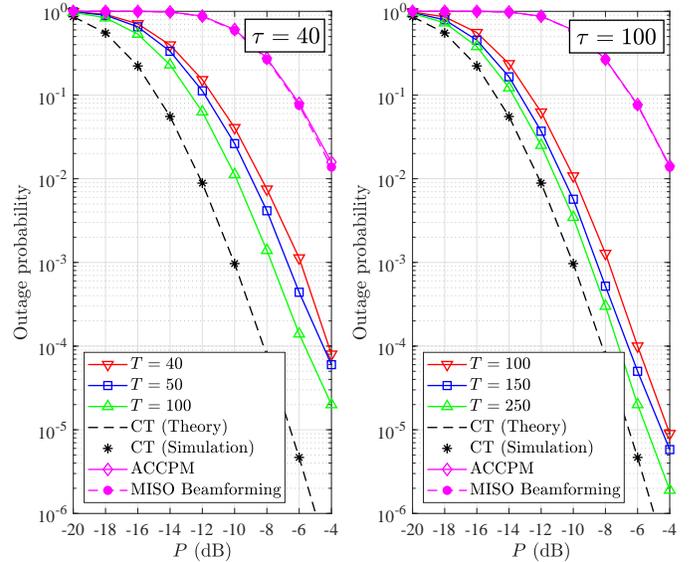}
	\caption{Outage probability versus $P$ for the OBF scheme; $M=10$, $\sigma_e^2 = 0$.}\label{fig4}
\end{figure}

Fig. \ref{fig5} illustrates the performance of the selection-based TD scheme with regards to the number of reflecting elements $M$ and the number of sub-surfaces $N$. Again, the theory (lines) and simulation (markers) are in agreement, which validates our analysis. In addition, we show the approximation of the TD scheme through the CLT, given by \eqref{td_app}. We can see that the approximation follows the behavior of the curves very well and it matches the simulation for high values of $M$ ($N=1$), which validates the consideration of Lemma \ref{lemma}. It is clear that the selection process improves the performance as $N$ increases, especially in the high SNR regime, where the scheme achieves full spatial diversity order. The performance of the CT scheme ($\sigma_e^2=0$) for $M=4$ and $M=10$ is also depicted. For a fair comparison, i.e. same number of active elements, we consider the TD scheme with $M=20, N=5$ and $M=20, N=2$. The CT scheme outperforms the TD scheme with $10$ active elements. However, we can see that the TD scheme provides significant gains when only $4$ elements are used. This shows how this scheme is both energy efficient and of low-complexity but can still provide significant performance gains.

\begin{figure}[t]\centering
	\includegraphics[width=\linewidth]{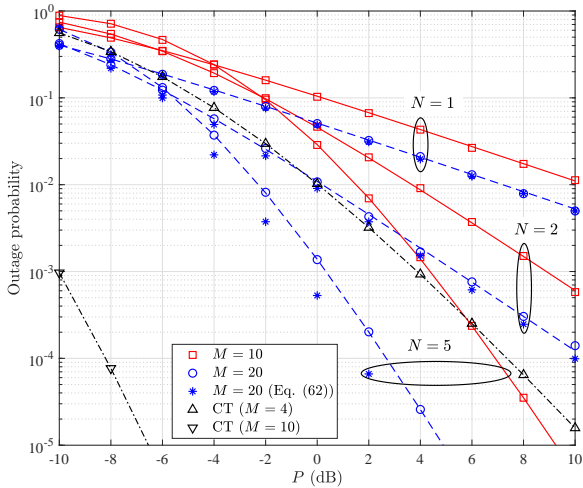}
	\caption{Outage probability versus $P$ for the TD scheme; $\sigma_e^2 = 0$.}\label{fig5}
\end{figure}

\begin{figure}[t]\centering
	\includegraphics[width=\linewidth]{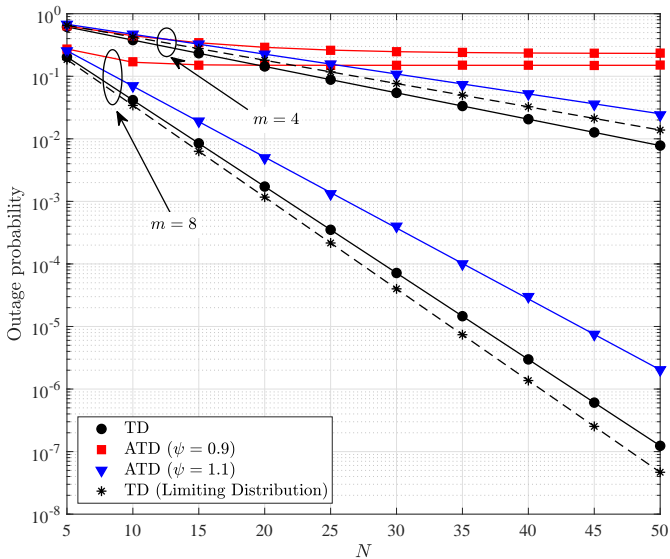}
	\caption{Outage probability versus $N$ for the selection-based schemes; $P = -10$ dB.}\label{fig6}
\end{figure}

Fig. \ref{fig6} depicts the outage probability for the selection-based schemes, in terms of the number of sub-surfaces $N$. In contrast to the other cases, the performance of the ATD scheme with $\psi = 0.9$, diminishes as $N$ increases. Since $\psi = 0.9 < \rho$, the IRS could select a sub-surface with achieved SNR greater than $\psi$ but lower than $\rho$. As $N$ increases, this selection is more likely and so the outage probability increases as well. On the other hand, when $\psi = 1.1 > \rho$, the outage probability decreases with $N$, since a selection in this case implies that the destination will not be in outage. These observations can also be derived from the analytical expression in Proposition \ref{prop_atd}. It is important to note that, even though the TD scheme outperforms the ATD scheme, it's implementation may require more bits of feedback. Fig. \ref{fig6} also shows the performance of TD using the limiting distribution. Despite deriving the limiting distribution using Lemma \ref{lemma}, we can see that it still describes the system's behavior very well.

Finally, Fig. \ref{fig7} shows the energy efficiency of the proposed schemes for different values of $M$ as well as for $P = -10$ dB (left sub-figure) and $P = 0$ dB (right sub-figure). The first main observation, is that the energy efficiency initially increases with $M$ but, after a certain value of $M$, it starts to decrease. This is expected, since the rate grows logarithmically but the power consumption grows linearly with $M$. Secondly, the energy efficiency of the RRC scheme is the same as with the conventional case (Proposition \ref{prop1}). As shown in Section \ref{rrc}, on average the rates of the two scenarios are equal. We then compare the OBF, the TD and the CT schemes with training periods $\tau = 10$ and $\tau = 20$. The OBF scheme outperforms all other schemes since at each time instant, the destination experiences a higher channel gain and gets closer to the beamforming gain and is not affected by a rate loss factor (see Section \ref{comp}). It is important to note here that our analytical approach for the OBF scheme provides a close approximation. The TD scheme has a smaller energy efficiency for small values of $M$, compared to the other cases. However, as $M$ increases, the TD scheme becomes more energy efficient than the RRC and CT schemes, due to the fact that it activates a fraction of the available elements at the IRS. Therefore, the energy efficiency of the TD scheme will start to decrease at larger values of $M$ compared to the other schemes; this is clearly more evident for the case $P=0$ dB.

\begin{figure}[t]\centering
	\includegraphics[width=\linewidth]{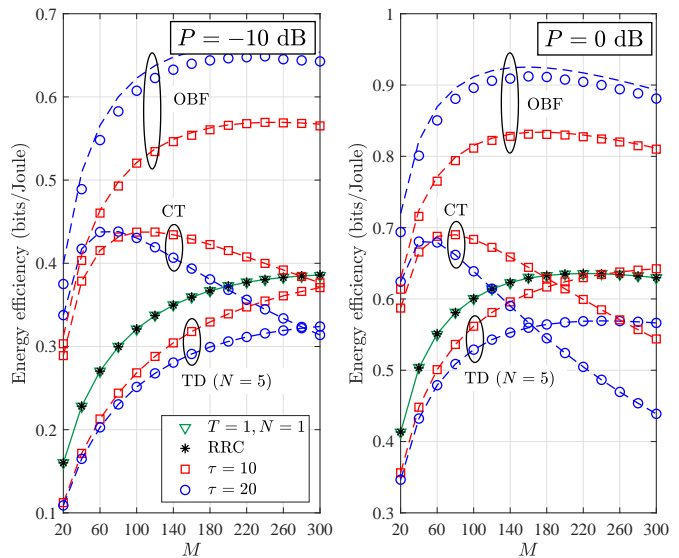}
	\caption{Energy efficiency versus number of elements $M$; $T = 500$, $P_o = 0$ dB, $\kappa = 0.6$, $\Delta = 0.1$.}\label{fig7}
\end{figure}

\section{Conclusion}\label{conclusion}
In this paper, we presented an analytical framework for random rotation-based IRS-aided communications and presented four low-complexity schemes that do not require instantaneous knowledge of any CSI. In particular, we proposed two coding-based schemes, which produce a time-varying channel through time-varying random rotations. Moreover, we proposed two selection-based schemes, which activate a partition of the IRS elements based on received signal power at the destination. Analytical expressions were derived for the outage probability and energy efficiency of all the proposed schemes. Moreover, the diversity order together with the coding gain achieved by each scheme was provided. Our results demonstrated that the proposed schemes provide significant performance gains compared to the conventional case, whilst keeping the complexity low and the energy efficiency high.

\appendices
\section{Proof of Proposition \ref{prop1}}\label{prop1_prf}
Since $T = 1$, the channel gain from the $M$ elements of the IRS is $H_1 = \Big\lvert \sum_{i=1}^M h_i g_i \exp(j\phi_{1,i})\Big\rvert^2$. Due to the random rotations, the phases have no effect on the channel gain. In other words, $H_1$ is statistically equivalent to $\Big\lvert \sum_{i=1}^M h_i g_i\Big\rvert^2$. Therefore, the outage probability when $T = 1$ can be evaluated as follows
\begin{align}
\Pi(\rho,1) &= \PP\{\log_2(1+\gamma_1) < \rho\} = \PP\left\{\frac{P}{\sigma_n^2} H_1 < 2^\rho-1\right\}\nonumber\\
&=\PP\left\{\Bigg\lvert \sum_{i=1}^M h_i g_i\Bigg\rvert^2 < \frac{\theta\sigma_n^2}{P}\right\},
\end{align}
where we defined $\theta \triangleq 2^\rho-1$. The final expression is then derived by using
\begin{align}\label{cdf}
F_{H_1}(x) = 1-\frac{2}{\Gamma(M)} \left(\frac{x}{\sigma^2}\right)^{M/2} K_M\left(2 \sqrt{\frac{x}{\sigma^2}}\right),
\end{align}
which is the CDF of $H_1$ \cite{JDG}.

\section{Proof of Theorem \ref{thm_rrc1}}\label{thm_rrc1_prf}
\begin{figure}[t]\centering
  \includegraphics[width=\linewidth]{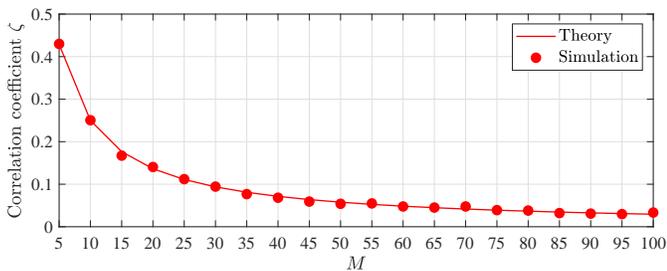}
  \caption{Correlation coefficient $\zeta$ versus $M$.}\label{fig8}
\end{figure}

The random variables $H_t = \Big\lvert\sum_{i=1}^M h_i g_i \exp(\jmath \phi_{t,i})\Big\rvert^2$ are correlated since the channel coefficients $h_i$ and $g_i$ remain constant over all $T$ transmissions. However, as $M$ increases, the correlation between the random variables $H_t$ decreases. In particular, the correlation coefficient $\zeta$ between $H_{t_1}$ and $H_{t_2}$, $t_1 \neq t_2$, is given by
\begin{align}\label{corr_coeff}
\zeta = \frac{\E\{H_{t_1} H_{t_2}\}-\E\{H_{t_1}\}\E\{H_{t_1}\}}{\sigma_{H_{t_1}}\sigma_{H_{t_1}}} = \frac{3}{M+2},
\end{align}
where $\sigma_{H_i} = \sqrt{\E\{H_i^2\}-\E\{H_i\}^2}$, $\E\{H_i\} = \sigma^2 M$, $\E\{H_i^2\} = 2\sigma^4 M(M+1)$ and  $\E\{H_i H_j\} = \sigma^4M(M+3)$. It is clear that for $M\to\infty$ we have $\zeta \to 0$; this is also depicted in Fig. \ref{fig5}. By taking this into consideration, we can evaluate the outage probability as follows
\begin{align}
&\Pi_{\rm RRC}^{\rm IND}(\rho,T) = \PP\left\{\frac{1}{T} \sum_{t=1}^T \log_2\left(1+\gamma_t\right) < \rho\right\}\nonumber\\
&= \PP\left\{\log_2 \prod_{t=1}^T \left(1+\gamma_t\right) < T\rho\right\}\nonumber\\
&= \E_{H_k}\Bigg\{F_{H_1}\left(\frac{\sigma_n^2}{P}\left(\frac{2^{T\rho}}{\prod_{t=2}^T \left(1+PH_t/\sigma_n^2\right)}-1\right)\right)\Bigg\},
\end{align}
which follows by the logarithmic identity $\log_2(x)+\log_2(y)=\log_2(xy)$, solving for $H_1$ and using the CDF of $H_1$ given by \eqref{cdf}. Since the random variables $H_t$ are assumed to be independent, we have
\begin{align}
&\Pi_{\rm RRC}^{\rm IND}(\rho,T) = \int_{z_T} \cdots \int_{z_2} \prod_{t=2}^T f_{H_t}(z_t)\nonumber\\
&\times F_{H_1} \Bigg(\frac{\sigma_n^2}{P}\Bigg(\frac{2^{T\rho}}{\prod_{t=2}^T \left(1+Pz_t/\sigma_n^2\right)}-1\Bigg) \Bigg) dz_2\cdots dz_T,
\end{align}
where $f_{H_t}(z_t)$ is the PDF of $H_t$ given by the derivative of \eqref{cdf}. The integration limits are evaluated by considering the inequality
\begin{align}
\frac{2^{T\rho}}{\prod_{t=2}^T \left(1+Pz_t/\sigma_n^2\right)}-1 > 0,
\end{align}
sequentially for each $z_t$. Finally, the transformation $w_t \to 1+Pz_t/\sigma_n^2$ provides the final expression.

\section{Proof of Lemma \ref{lemma}}\label{lemma_prf}
Let $W_t = \sum_{i=1}^M h_i g_i \exp(\jmath \phi_{t,i}) = X_t + \jmath Y_t$, where $X_t = \sum_{i=1}^M |h_i||g_i| \cos(\phi_{t,i} + \chi_i)$ and $Y_t = \jmath\sum_{i=1}^M |h_i||g_i|\sin(\phi_{t,i}+\chi_i)$. We consider the case of large $M$, which implies that $W_t$ are independent and identically distributed (see Appendix \ref{thm_rrc1_prf}). The mean and variance of each summand are $0$ and $\sigma^2/2$, respectively, which follows from $\E\{\cos(\phi_{t,i}+\chi_i)\} = \E\{\sin(\phi_{t,i}+\chi_i)\} = 0$, $\E\{|h_i|^2\}\E\{|g_i|^2\} = \sigma_h^2\sigma_g^2 = \sigma^2$ and $\E\{\cos^2(\phi_{t,i}+\chi_i)\} = \E\{\sin^2(\phi_{t,i}+\chi_i)\} = 1/2$. Thus, by applying the CLT, we have that $X_t \sim \mathcal{N}(0,\sigma^2 M/2)$ and $Y_t \sim \mathcal{N}(0,\sigma^2 M/2)$. It is straightforward to show that $X_t$ and $Y_t$ are uncorrelated and jointly Gaussian. As such, $W_t$ converges in distribution to a complex Gaussian random variable and so $H_t = |W_t|^2$ is exponentially distributed with parameter $1/(\sigma^2 M)$.

\section{Diversity of RRC Scheme}\label{div_rrc}
Firstly, assume that $M > T$. Then, we can use Theorems \ref{thm_rrc1} and \ref{thm_rrc2} to derive the diversity order. By employing the approximation $K_M(x) \approx \Gamma(M)2^{M-1}/x^M$ for $x \approx 0$ \cite{ISG} in Theorem \ref{thm_rrc1}, we can see that the outage probability $\Pi_{\rm RRC}^{\rm IND}(\rho,T)$ reduces to zero. This is because the convergence to the diversity order for cascaded channels is very slow and is observed for very high SNR values \cite{MURAT}. However, using the approximated expression in Theorem \ref{thm_rrc2} and the fact that $\exp(-x) \approx 1-x$ for $x\approx 0$, we have 
\begin{align}
\lim\limits_{P\to\infty} \Pi_{\rm RRC}^{\rm CLT}(\rho,T) &\approx \left(\frac{\sigma_n^2}{MP}\right)^T \!\int_1^{\xi_T}\! \cdots \int_1^{\xi_2} \left(\frac{2^{\rho T}}{\prod_{t=2}^T w_t}-1\right) \nonumber\\
&\quad\times \prod_{t=2}^T \left(1-\frac{\sigma_n^2}{MP}(w_t-1)\right) dw_2\cdots dw_T\nonumber\\
&\quad\to O(1/P^T),\nonumber
\end{align}
where it is clear that the expansion of the second product will be a sum of $2^{T-1}$ terms, out of which only the term equal to one will not contain $1/P$. Therefore, as the smallest order term will dominate the others, it follows that the RRC scheme achieves spatial diversity of order $T$ with coding gain $G_{\rm RRC}$ equal to
\begin{align}
G_{\rm RRC} &= \left(\frac{\sigma_n^2}{M}\right)^T \int_1^{\xi_T} \cdots \int_1^{\xi_2} \left(\frac{2^{\rho T}}{\prod_{t=2}^T w_t}-1\right) dw_2\cdots dw_T\nonumber\\
&=\left(\frac{\sigma_n^2}{M}\right)^T (-1)^T\left(1 - 2^{\rho T} \sum_{t=0}^{T-1} \frac{(-1)^{t}}{t!} \log^t(2^{\rho T})\right),
\end{align}
which follows by evaluating the $(T-1)$-fold integral and after some trivial algebraic manipulations.

Now, consider the case $T > M$. In fact, assume that $T \to \infty$. In this case, we have
\begin{align}
&\lim_{T\to\infty} \frac{1}{T} \sum_{t=1}^T \log_2\left(1+\gamma_t\right)\nonumber\\
&= \E_{\bm{\phi}}\left\{\log_2\left(1+\frac{P}{\sigma_n^2} \Bigg\lvert \sum_{i=1}^M h_i g_i \exp(\jmath\phi_i)\Bigg\rvert^2\right)\right\},
\end{align}
where $\bm{\phi} = [\phi_1 ~ \phi_2 ~ \cdots ~ \phi_M]$. Moreover, at the high SNR regime, we have \cite{SHENG3}
\begin{align}
&\PP\left\{\E_{\bm{\phi}}\left\{\log_2\left(1+\frac{P}{\sigma_n^2} \Bigg\lvert \sum_{i=1}^M h_i g_i \exp(\jmath\phi_i)\Bigg\rvert^2\right)\right\} < \rho\right\}\nonumber\\
&\doteq \PP\left\{\log_2\left(1+\frac{P}{\sigma_n^2} \sum_{i=1}^M |h_i|^2 |g_i|^2\right) < \rho \right\},
\end{align}
where the relation $a \doteq b^c$ means $\lim_{b\to 0} \frac{\log a}{\log b} = c$ \cite{SHENG3}. This allows the use of $\PP\left\{\log_2\left(1+\frac{P}{\sigma_n^2} \sum_{i=1}^M |h_i|^2 |g_i|^2\right) < \rho \right\}$ to derive the scheme's diversity order. As such, we have
\begin{align}
&\PP\left\{\sum_{i=1}^M |h_i|^2 |g_i|^2 < (2^\rho-1)\frac{\sigma_n^2}{P} \right\}\nonumber\\
&= \frac{1}{\pi} \int_0^\infty \Im\left\{\phi(t) \sum_{i=M}^\infty (-1)^{i+1}\left(\jmath \frac{\theta\sigma_n^2}{P}\right)^i \frac{t^{i-1}}{i!}\right\} dt,
\end{align}
which follows from the Gil-Pelaez inversion theorem \cite{GP} and the Taylor series expansion of the exponential function (see Appendix \ref{propc_prf} for details); $\phi(t) = (\int_0^\infty\exp(\jmath t h) f(h) dh)^M$ is the characteristic function of $\sum_{i=1}^M |h_i|^2 |g_i|^2$ and $f(h) = 2K_0(2\sqrt{h})$ is the PDF of $|h|^2 |g|^2$. Then, for $P \to \infty$, we have
\begin{align}
&\lim_{P\to\infty} \frac{1}{\pi} \int_0^\infty \Im\left\{\phi(t) \sum_{i=M}^\infty (-1)^{i+1}\left(\jmath \frac{\theta\sigma_n^2}{P}\right)^i \frac{t^{i-1}}{i!}\right\} dt \nonumber\\
&\approx \frac{1}{\pi} \int_0^\infty \Im\left\{\phi(t) (-1)^{M+1}\left(\jmath \frac{\theta\sigma_n^2}{P}\right)^M \frac{t^{M-1}}{M!}\right\} dt\nonumber\\
&\to O(1/P^M),
\end{align}
which shows that the diversity order is $M$. 

\section{Proof of Proposition \ref{prop_rtd}}\label{prop_rtd_prf}
Given the SNR ordering $\gamma_{(1)} \geq \gamma_{(2)} \geq \dots \geq \gamma_{(N)}$, the expected rate of the highest received SNR $\gamma_{(1)}$ at the destination can be derived as
\begin{align}
\E_{\gamma_{(1)}}\{\log_2(1+\gamma_{(1)})\} &= \E_{H_1}\left\{\log_2\left(1+\frac{P}{\sigma_n^2}H_{(1)}\right)\right\}\nonumber\\
&=\int_0^\infty \log_2\left(1+\frac{P}{\sigma_n^2} h\right) p_{H_{(1)}}(h) dh,\label{er}
\end{align}
where $p_{H_{(1)}}$ is the probability distribution function (PDF) of the largest order statistic $H_{(1)}$ given by \cite{SLIM}
\begin{align}
p_{H_{(1)}}(h) = N F_H(h)^{N-1} f_H(h).
\end{align}
By replacing the above PDF in \eqref{er}, completes the proof.

\section{Proof of Proposition \ref{prop_atd}}\label{prop_atd_prf}
Two cases need to be considered, namely, $\rho \leq \psi$ and $\rho > \psi$. In the former case, outage occurs in the event of $N-1$ sub-surfaces not satisfying the selection criterion. As the events are mutually exclusive, we have
\begin{align}
\Pi_{\rm ATD}(\rho,\psi \,|\, \rho \leq \psi) &= \PP\{\log_2(1+\gamma_1) < \psi\}\cdots\nonumber\\
&\quad\times\PP\{\log_2(1+\gamma_{N-1}) < \psi\}\nonumber\\
&\quad\times\PP(\log_2(1+\gamma_N) < \rho).
\end{align}
In the other case, outage occurs when the $i$-th sub-surface, $1 \leq i \leq N-1$, satisfies the selection criterion but not the outage threshold or when $N-1$ sub-surfaces are not selected and the $N$-th is in outage. In mathematical terms,
\begin{align}
&\Pi_{\rm ATD}(\rho,\psi \,|\, \rho > \psi) = \PP\{\psi < \log_2(1\!+\!\gamma_1) < \rho\}\nonumber\\
&+ \PP\{\log_2(1\!+\!\gamma_1) < \psi\}\PP\{\psi < \log_2(1\!+\!\gamma_2) < \rho\}\nonumber\\
&+ \cdots + \PP\{\log_2(1+\gamma_1) < \psi\} \cdots\PP\{\log_2(1+\gamma_{N-2}) < \psi\}\nonumber\\
&\hspace{4cm}\times\PP\{\psi < \log_2(1+\gamma_{N-1}) < \rho\}\nonumber\\
&+\PP\{\log_2(1+\gamma_1) < \psi\} \cdots\PP\{\log_2(1+\gamma_{N-1}) < \psi\}\nonumber\\
&\hspace{4cm}\times\PP(\log_2(1+\gamma_N) < \rho).
\end{align}
By using Proposition \ref{prop1} for the outage probability of each event, the result follows.

\section{Diversity of Selection-based Schemes}\label{div_sbs}
Once again, we use the approximated expression in Theorem \ref{thm_rrc2} with $T=1$. For the TD scheme, we have
\begin{align}\label{td_app}
\Pi_{\rm TD}(\rho) \approx \left(\Pi_{\rm RRC}^{\rm CLT}(\rho,1)\right)^N = \left(1-\exp\left(-\frac{\theta\sigma_n^2}{MP}\right)\right)^N,
\end{align}
and so
\begin{align}
\lim\limits_{P\to\infty} \left(\Pi_{\rm RRC}^{\rm CLT}(\rho,1)\right)^N \approx \left(\frac{\theta\sigma_n^2}{MP}\right)^N \to O(1/P^N),
\end{align}
which follows from $\exp(-x) \approx 1-x$ for $x\approx 0$. Thus, the selection-based TD scheme achieves diversity order $N$, with a coding gain equal to \eqref{ag_td}.

\section{Proof of Theorem \ref{propc}}\label{propc_prf}

In the case of CT, the phases are aligned and so the channel gain is given by $\left(\sum_{i=1}^M \lvert h_i g_i \rvert\right)^2$. Now, the PDF for the product of two Rayleigh random variables $\lvert h \rvert$ and $\lvert g \rvert$ is given by $f_{\lvert h g \rvert}(x) = \frac{4 x}{\sigma^2} K_0\left(\frac{2x}{\sigma}\right)$, with $\sigma^2=\sigma_h^2\sigma_g^2$ \cite{CP}. By taking into account the errors induced by the channel estimation, the source estimates $\widehat{h g}$ with PDF
\begin{align}\label{ct_pdf}
f_{\lvert \widehat{h g} \rvert}(x) = \frac{4 x}{\sigma^2(1-\sigma_e^2)} K_0\left(\frac{2x}{\sqrt{\sigma^2(1-\sigma_e^2)}}\right),
\end{align}
where $\sigma_e^2$ captures the channel estimation accuracy of the cascaded channel \cite{GOLD}; we consider $\sigma_e^2 = 1/(1+\tau P_o/M)$ \cite{HASS}, where $P_o$ is the transmit power of the pilot. It follows that the SNR is given by
\begin{align}\label{snr_ct}
\gamma_{\rm CT} = \frac{P}{\sigma_n^2 + \sigma_e^2} \Bigg(\sum_{i=1}^M \lvert \widehat{h_i g_i} \rvert\Bigg)^2.
\end{align}
Then, the outage probability can be evaluated as
\begin{align}
\Pi_{\rm CT}(\rho) &= \PP\left\{\Bigg(\sum_{i=1}^M \lvert \widehat{h_i g_i} \rvert\Bigg)^2 < \frac{\theta(\sigma_n^2+\sigma_e^2)}{P}\right\}\nonumber\\
&= \PP\left\{\sum_{i=1}^M \lvert \widehat{h_i g_i} \rvert < \sqrt{\frac{\theta(\sigma_n^2+\sigma_e^2)}{P}}\right\}.
\end{align}
Let $s = \sigma^2(1-\sigma_e^2)$. Then, by using the above PDF, the characteristic function $\phi_{\lvert \widehat{h g} \rvert}(t)$ of $\lvert \widehat{h g} \rvert$ can be evaluated as
\begin{align}
\phi_{\lvert \widehat{h g} \rvert}(t) &= \E_{\lvert \widehat{h g} \rvert}\{\exp(\jmath t \lvert \widehat{h g} \rvert)\}\nonumber\\
&= \frac{4}{s} \int_0^\infty \exp(\jmath t x) x K_0\left(\frac{2x}{\sqrt{s}}\right) dx\nonumber\\
&= \frac{4\sqrt{s t^2+4} + 2\jmath \sqrt{s} t\left(\pi + 2\jmath \sinh^{-1}\left(\frac{\sqrt{s} t}{2}\right)\right)}{\left(st^2+4\right)^{3/2}},\label{cf}
\end{align}
which follows with the help of \cite[6.624-1]{ISG} and the fact that $\log(\jmath) = \jmath\pi/2$ and $\log(x+\sqrt{x^2+1}) = \sinh^{-1}(x)$ \cite{ISG}, where $\sinh(\cdot)$ is the hyperbolic sine function.

\begin{figure}[t]\centering
	\includegraphics[width=\linewidth]{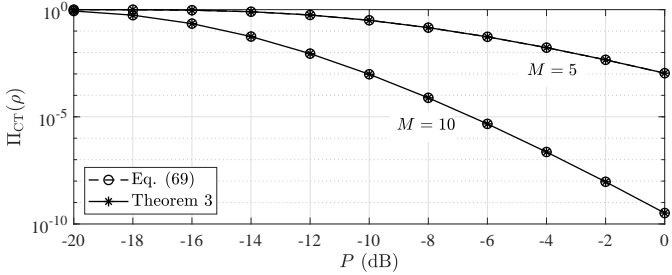}
	\caption{$\Pi_{\rm CT}(\rho)$ versus $P$ with $\theta = 1$, $\sigma_n^2 = 0$ dB; Theorem \ref{propc} agrees with Eq. \eqref{ctsum}, which validates our approach.}\label{fig9}
\end{figure}
	
Using the Gil-Pelaez inversion theorem \cite{GP}, we can obtain $\Pi_{\rm CT}(\rho)$ as follows
\begin{align}\label{pct}
\Pi_{\rm CT}(\rho) = \frac{1}{2} - \frac{1}{\pi}\int_0^\infty\!\!\Im\Bigg\{\!\!\exp\left(-\jmath t\sqrt{\frac{\theta(\sigma_n^2+\sigma_e^2)}{P}}\right) \phi(t)\!\Bigg\}\frac{dt}{t} ,
\end{align}
where $\phi(t)$ is the characteristic function of $\sum_{i=1}^M \lvert h_i \rvert \lvert g_i \rvert$, i.e. the sum of $M$ independent products of Rayleigh random variables. Hence, it follows that $\phi(t) = \phi_{\lvert \widehat{h g} \rvert}(t)^M$, where $\phi_{\lvert \widehat{h g} \rvert}(t)$ is given by \eqref{cf}. By applying a Taylor series expansion to the exponential function in \eqref{pct} and using the fact that $\int_0^\infty \frac{1}{t}\Im\{\phi(t)\} dt = \pi/2$, we have
\begin{align}\label{ctsum}
\Pi_{\rm CT}(\rho) = - \frac{1}{\pi} \int_0^\infty \Im\Bigg\{\sum_{i=1}^\infty\! \left(-\jmath t \sqrt{\frac{\theta(\sigma_n^2+\sigma_e^2)}{P}}\right)^i \frac{\phi(t)}{i!} \Bigg\} \frac{dt}{t}.
\end{align}
The final result follows by taking into account that the first $2M-1$ terms of the above sum are zero; this is validated by Fig. \ref{fig9}.

\end{document}